\documentclass[12pt]{amsart}

\usepackage{amssymb,latexsym,amsmath,extarrows, mathrsfs, amsthm,color,amsrefs}
\usepackage{graphicx}
\usepackage{amsmath}

\usepackage[margin=1in, centering]{geometry}
\usepackage[bookmarksnumbered, colorlinks, plainpages]{hyperref}
\usepackage{hyperref} 
\hypersetup{
    colorlinks=true,       
    linkcolor=blue,          
    citecolor=magenta,        
    filecolor=magenta,      
    urlcolor=cyan           
}

\usepackage{mathtools}
\mathtoolsset{showonlyrefs,showmanualtags}

\newtheorem{theorem}{Theorem}[section]

\newtheorem{lemma}[theorem]{Lemma}

\newtheorem{remark}[theorem]{Remark}

\newcommand{\C}{\mathbb C}
\newcommand{\Z}{\mathbb Z}
\newcommand{\T}{\mathbb T}

\newcommand{\R}{\mathbb R}

\newcommand{\ol}{\overline}
\newcommand{\D}{\mathbb D}

\title[Quantum walk in electric fields]{Interval spectrum for electric quantum walk and related skew-shift CMV matrices}
\author{}
\begin{document}

\author[F.\ Yang]{Fan Yang}
\address{Department of Mathematics \\ Louisiana State University  \\  Baton Rouge, LA 70803, USA}
\email{yangf@lsu.edu}

\thanks{
}

\begin{abstract} 
We show that for a family of quantum walk models with electric fields, the spectrum is the unit circle for any irrational field. The result also holds for the associated CMV matrices defined by skew-shifts. Generalizations to CMV matrices with skew-shifts on higher dimensional torus are also obtained.
\end{abstract}

\maketitle

\section{Introduction}

The quantum walk model is a quantum mechanical analogue of the classical random walk. It can also be used to describe quantum simulation of a particle on a lattice or graph. 
This model has demonstrated a wide range of applications across various fields, including mathematical physics and quantum information, attracting significant attention over the past few decades, see e.g. \cite{VA12,Kon08,YG,SC, WM,ARS}. 

It was discovered by Cantero, Moral, Grünbaum and Velázquez \cite{CGMV}  that quantum walk models can be reduced to CMV matrices, a family of extensively studied unitary operators that have further connections to orthogonal polynomials on the unit circle (OPUC) \cite{SOPUC,5years}. This connection brings the tools that were largely developed in the spectral theory community into the study of quantum walks.

The study of spectrum structure is an important topic in spectral theory and mathematical physics. One class of operator that was studied widely is quasi-periodic Schr\"odinger operators. For example, the famous Ten Martini problem, named after Kac and Simon, conjectured Cantor spectrum for the almost Mathieu operators (quasi-periodic operator with cosine potential) for any irrational magnetic flux.
The complete proof of this conjecture was given by Avila-Jitomirskaya \cite{AJ1}, with earlier important advances by Bellissard-Simon \cite{BS}, Choi-Elliot-Yui \cite{CEY}, Last \cite{La}, Puig \cite{Puig} and Avila-Krikorian \cite{AK}.
Goldstein and Schlag proved Cantor spectrum for quasi-periodic operators with general analytic potentials defined on $\T$ in the positive Lyapunov exponent regime for a.e. magnetic flux \cite{GS}.
See \cite{DGY,Handry,BHJ,CFO,SSY,BBL,KPS2} for Cantor spectrum results for some other quasi-periodic models.
Though Cantor spectrum is generally expected for many 1D single-frequency quasi-periodic operators, there are some exceptions.
One such known example is the Maryland model, with a $\tan$ potential, for which the spectrum was proved to be the whole real line \cite{SimonMaryland}.
The monotonicity of the $\tan$ potential also leads to other unique features, e.g. the lack of phase resonance, see \cite{JL,JY,HJY}.
For multi-frequency or higher dimensional operators, interval spectrum is conjectured.
A significant advancement in the multi-frequency case was made by Goldstein-Schlag-Voda in \cite{GSV}, where they solved this problem up to the spectral edges for large coupling constants. 
For higher dimensional continuum quasi-periodic operators, a recent breakthrough of Karpeshina-Parnovski-Shterenberg \cite{KPS} established the Bethe-Sommerfeld conjecture for generic potentials, hence proving finiteness of the number of gaps.

In this short note, we show that a quantum walk model in electric field, studied actively by physicists in recent years \cite{CRWAGW,CWloc,GASWWMA, BHSVR, APP}, does not have spectral gaps. 
Our proof is elementary but the result provides a new example of non-Cantor spectrum for a 1D single-frequency quasi-periodic operators. 
Another interesting feature of the model is that it can be reduced to a CMV matrix defined by skew-shift. 
The skew-shift operators, see the dynamics in \eqref{def:sk}, is a different family of ergodic operators, hosting a number of major open problems. One of the conjectures is that skew-shift models do not have spectral gaps, see Chapter 15 of \cite{Bourgainbook}. 
The results in this note, Theorem \ref{thm:spec1} and \ref{thm:spec2}, provide examples for which this skew-shift conjecture is true.
See also \cite{BGS,B1,K1,K2,Hconti,HLS} for some related works concerning the skew-shift dynamics.

Next let us introduce the model and our results in details.
We study the discrete time quantum walk of one particle on $\Z$ with two dimensional internal degree of freedom with external electric fields.
Compare to quantum walk with external magnetic field, for which there is a vast literature of mathematical studies recently, see e.g. \cite{CFO,qw,AD}, quantum walk with electric field is less explored.

The quantum walk model acts on $\ell^2(\Z)\otimes \C^2=:\mathcal{H}$ with the following basis:
\begin{align*}
\delta_n^{\pm}=\delta_n\otimes e_{\pm},\ \  n\in \Z,
\end{align*} 
where $\{\delta_n\}$ is the standard basis of $\ell^2(\Z)$ and $\{e_+=(1,0)^T,\ \ e_-=(0,1)^T\}$ is the standard basis of $\C^2$.
Let $C\in U(2)$ be a coin matrix, defined by
\begin{align}\label{def:C0}
C=e^{2\pi i\eta}
\left(\begin{matrix}
a & b\\
-b^* &a^*
\end{matrix}\right),
\end{align}
in which $|a|^2+|b|^2=1$ and  $\eta \in \T$.

Let $S$ be a conditional shift operator, defined as
\begin{align}\label{def:S}
S\delta_n^{\pm}=\delta_{n\pm 1}^{\pm}, 
\text{ or } (S\psi)_n^{\pm}=\psi_{n\mp 1}^{\pm}, \text{ for } \psi \in \ell^2(\Z) \otimes \C^2,
\end{align}
The evolution of the quantum walk for each step is described by a unitary operator, which can be represented as a product $CS$, i.e. a shift-coin quantum walk.
Here we consider the following electric walk $W$ acting on $\mathcal{H}$,
\begin{align}\label{def:W}
W_{\omega,\theta,\eta,a,b}:=Q_{\omega,\theta,\eta,a,b} S,
\end{align}
where $\omega\in \T$ (the discrete electric field), $\theta\in \T$ and  $Q_{\omega,\theta,\eta,a,b}$ acts coordinate-wise via a matrix multiplication $Q_{\omega,\theta,\eta,a,b,n}$ defined by 
\begin{align}\label{def:Q}
Q_{\omega,\theta,\eta,a,b,n}(\psi_n^-, \psi_n^+)^T=e^{2\pi i(\theta+n\omega)}C\, (\psi_n^-, \psi_n^+)^T.
\end{align}
In the following, we shall identify $\ell^2(\Z)\otimes \C^2 \to \ell^2(\Z)$ via
$$\delta_n^+\mapsto \delta_{2n+1},\ \ \delta_n^-\mapsto \delta_{2n}.$$
Under such identification, $W_{\omega,\theta,\eta,a,b}$ can be written as a five diagonal matrix, see below.

\begin{equation}\label{def:qw}
 W_{\omega,\theta,\eta,a,b}=
\left(\begin{matrix}
&\cdots\\
&b e^{2\pi i (\theta+\eta+n\omega)}  &\underline{0}  &0  &a e^{2\pi i (\theta+\eta+n\omega)} & & \\
&a^* e^{2\pi i (\theta+\eta+n\omega)}  &0 &0 &-b^* e^{2\pi i (\theta+\eta+n\omega)} & & \\
& & &b e^{2\pi i (\theta+\eta+(n+1)\omega)} &0 &0 &a e^{2\pi i (\theta+\eta+(n+1)\omega)} &\\
& & &a^* e^{2\pi i (\theta+\eta+(n+1)\omega)} &0 &0 &-b^* e^{2\pi i (\theta+\eta+(n+1)\omega)} &\\
& & & & & &\cdots
\end{matrix}\right),   
\end{equation}
in which $\underline{0}$ indicates the $(2n,2n)$ position of the matrix.

As we mentioned, a remarkable observation that quantum walk operators are unitarily equivalent to CMV matrices was made in \cite{CGMV}. Now we introduce the definition of CMV matrices.

Let $\D$ be the unit disk in $\C$, $\partial \D$ be the unit circle and $\{\alpha_n\}_{n\in \Z}\subset \D$ be the {\it Verblunsky coefficients}.
Define 
\begin{align}\label{def:Theta}
\Theta_n:=\left(\begin{matrix}
\ol{\alpha}_n &\rho_n\\
\rho_n &-\alpha_n
\end{matrix}\right),
\end{align}
where \begin{align}
    \rho_n=\sqrt{1-|\alpha_n|^2}.
\end{align}
Let
\begin{align}\label{def:LM}
\mathcal{L}:=\bigoplus \Theta_{2n}, \text{ and } \mathcal{M}:=\bigoplus \Theta_{2n+1},
\end{align}
where $\Theta_{2n}$ acts on $\ell^2(\{2n, 2n+1\})$, and $\Theta_{2n+1}$ acts on $\ell^2(\{2n+1, 2n+2\})$.
CMV matrices are defined by
\begin{align}\label{def:CMV}
\mathcal{E}:=\mathcal{L}\mathcal{M},
\end{align} which are uniquely determined by the Verblunsky coefficients and $\mathcal{E}$ has the following expression:
\begin{align*}
\mathcal{E}=\begin{bmatrix}
		\ddots & \ddots & \ddots & \ddots &&&&  \\
		& \overline{\alpha_{2j}\rho_{2j-1}} & -\overline{\alpha_{2j}}\alpha_{2j-1} & \overline{\alpha_{2j+1}}\rho_{2j} & \rho_{2j+1}\rho_{2j} &&&  \\
		& \overline{\rho_{2j}\rho_{2j-1}} & -\overline{\rho_{2j}}\alpha_{2j-1} & {-\overline{\alpha_{2j+1}}\alpha_{2j}} & -\rho_{2j+1} \alpha_{2j} &&&  \\
		&&  & \overline{\alpha_{2j+2}\rho_{2j+1}} & -\overline{\alpha_{2j+2}}\alpha_{2j+1} & \overline{\alpha_{2j+3}} \rho_{2j+2} & \rho_{2j+3}\rho_{2j+2} & \\
		&& & \overline{\rho_{2j+2}\rho_{2j+1}} & -\overline{\rho_{2j+2}}\alpha_{2j+1} & -\overline{\alpha_{2j+3}}\alpha_{2j+2} & -\rho_{2j+3}\alpha_{2j+2} &    \\
		&& && \ddots & \ddots & \ddots & \ddots &
	\end{bmatrix}.
\end{align*}

This is indeed an extended CMV matrix, as it is a whole-line operator and original CMV matrices are half-line operator.

The equivalence between this particular electric quantum walk model $W_{\omega,\theta,\eta, a, b}$ and CMV matrix was computed explicitly in \cite{CWloc}, see below.

\begin{lemma}\label{lem:CMV_skew}\cite{CWloc}
There exists a diagonal unitary matrix $\Lambda=\mathrm{diag}(e^{2\pi i\lambda_j})$ with $\lambda_j\in \R$, such that
\begin{align*}
\Lambda^{-1}W_{\omega,\theta,\eta,a,b}\Lambda=\mathcal{E}_{\omega,\theta,\eta,|a|,|b|},
\end{align*}
such that $\mathcal{E}_{\omega,\theta,\eta,|a|,|b|}$ is a CMV matrix with the following Verblunsky coefficients:
\begin{align}\label{def:alpha_n}
\alpha_{2n}=|b| e^{-2\pi i((n^2-n)\omega+2n(\theta+\eta))},\, \rho_{2n}=|a|, \text{ and } \alpha_{2n+1}=0,\, \rho_{2n+1}=1 \ \ \text{ for any } n\in \Z.
\end{align}
\end{lemma}
\begin{align*}
\mathcal{E}_{\omega,\theta,\eta,|a|,|b|}=\left(\begin{matrix}
&\cdots\\
&|b|e^{-2\pi i \tilde{\psi}_j}  &0^{\star}  &0  &|a| & & \\
&|a| &0 &0 &-|b|e^{2\pi i \tilde{\psi}_j} & & \\
& & &|b|e^{-2\pi i \tilde{\psi}_{j+1}} &0 &0 &|a| \\
& & &|a| &0 &0 &-|b|e^{2\pi i \tilde{\psi}_{j+1}} \\
& & & & & &\cdots
\end{matrix}\right),
\end{align*}
in which $\tilde{\psi}_j=2j(\theta+\eta)+(j^2-j)\omega$.

\begin{remark}
Note that the appearance of $n^2\omega$ in the Verbluksky coefficients. This is one interesting feature of the electric field model, which differs from the quantum walk under magnetic field.
The $n^2\omega$ indicates the associated CMV matrix is generated by the skew-shift dynamics on $\T^2$, defined by:
\begin{align}\label{def:sk}
\, \, \, T_{2, \omega}(x_1,x_2)=(x_1+\omega, x_2+x_1), \text{ and } (T_{2, \omega})^n(x_1,x_2)=(x_1+n\omega, x_2+nx_1+{n\choose 2}\omega).
\end{align}
In fact, let $\mathcal{P}(x_1,x_2)=x_2$ be the projection: $\T^2 \rightarrow \T$ on the second component, then
\begin{align*}
\mathcal{P}[(T_{2, \omega})^n ( \theta+\eta, 0)]= n(\theta+\eta)+{n\choose 2}\omega
\end{align*}
and hence the coefficients $\alpha_{2n}$ in the associated CMV matrix $\mathcal{E}_{\omega,\theta,\eta,|a|,|b|}$ is given by 
\begin{align}
    \alpha_{2n}=|b|e^{-4\pi i \mathcal{P}[(T_{2, \omega})^n(\theta+\eta,0)]}.
\end{align}
\end{remark}

Let $\sigma(A)$ be the spectrum of an operator $A$. Our first result is the following:
\begin{theorem}\label{thm:spec1}
For any irrational $\omega$, we have
\begin{align*}
\sigma(W_{\omega,\theta,\eta,a,b})=\sigma(\mathcal{E}_{\omega,\theta,\eta,|a|,|b|})=\partial \D,
\end{align*}
where $\D$ is the unit disk in $\C$, for any $\theta \in \T,\eta \in \T$ and $a \in \C,b \in \C$ such that $|a|^2+|b|^2=1$.
\end{theorem}

\begin{remark}
    Since $W$ and $\mathcal{E}$ are unitary operators, it is obvious that their spectra are contained in $\partial \D$, like other quantum walk models. Thus one only needs to prove the inverse direction, which is presented in Section \ref{sec:spec1}. It is also a known fact that 
    \begin{align}
        \sigma(W_{\omega,\theta,\eta,a,b})=e^{2\pi i\theta}\cdot \sigma(W_{\omega,0,\eta,a,b}).
    \end{align} 
    But we could not find the description of $\sigma(W_{\omega,\theta,\eta,a,b})$ in the literature, this is one of the motivations why we write up this note. The proof of Theorem \ref{thm:spec1} is elementary.
\end{remark}

Indeed, the same result also holds for a more general family of CMV matrices defined by iterated (higher-dimensional) skew shift on $\T^{d}$, $d\geq 2$ as follows. 
Let $T_{d, \omega}$ be the iterated skew-shift on $\T^{d}$ defined as
\begin{align}\label{def:Skew}
T_{d, \omega}(\vec{x})=T_{d,\omega}(x_1,x_2,...,x_d)=(x_1+\omega, x_2+x_1,..., x_d+x_{d-1}).
\end{align}

Let $\mathcal{P}_{d}: \T^{d}\mapsto \T$ be the projection onto the $d$-coordinate:
\begin{align}\label{def:Pd}
\mathcal{P}_{d}(x_1,x_2,...,x_{d})=x_{d}.
\end{align}
Let $a,b\in [0,1]$ satisfying $a^2+b^2=1$.
For each $n\in \Z$, $\tilde{\alpha}_n$ is defined as
\begin{align}\label{def:alphan_2}
\alpha_{2n}(\vec{x}):=b e^{2\pi i\mathcal{P}_{d}(T_{d, \omega}^n(\vec{x}))}, \rho_{2n}\equiv a, \text{ and } \alpha_{2n+1}\equiv 0, \rho_{2n+1}\equiv 1.
\end{align}
Let $\mathcal{E}^d_{\omega,a,b}(\vec{x})$ be the associated CMV matrix determined by $\tilde{\alpha}_n$ and $\rho_n$ as its Verblunsky coefficients. 
It is known that for irrational $\omega$, $T_{d, \omega}$ is minimal \cite{DS}, and hence 
the spectrum $\sigma(\mathcal{E}^d_{\omega,a,b}(\vec{x}))=:\sigma(\mathcal{E}^d_{\omega,a,b})$ is independent of $\vec{x}$, see e.g. \cite{BIST,CFKS,BLLS}.
For this model, we have:

\begin{theorem}\label{thm:spec2}
For any irrational $\omega$, and $a,b\in [0,1]$ such that $a^2+b^2=1$. we have
\begin{align*}
\sigma(\mathcal{E}^d_{\omega,a,b})=\partial \D.
\end{align*}
\end{theorem}

The rest of the paper is organized as follows: Section 2 provides preliminaries for Theorem \ref{thm:spec2}. We present the short proof of Theorem \ref{thm:spec1} in Sec. \ref{sec:spec1} and the proof of Theorem \ref{thm:spec2} in Sec. \ref{sec:spec2}.

\section{Preliminaries}
\subsection{Some facts about combinatorial numbers}
\

For $k\geq 0$, the following is the combinatorial number, with the convention that $0!=1$,
\begin{align}\label{eq:comb_pos}
    {n \choose k}=
       \begin{cases}
       & \frac{n!}{k! (n-k)!}, \text{ if } n\geq k\geq 0,\\
       & 0, \text{ if } 0\leq n<k.
       \end{cases}
     \end{align} 
For $n\leq -1$ and $k \geq 1$, the following definition is standard:
\begin{align}\label{eq:comb_neg}
    {n\choose k}=\frac{n\cdots (n-k+1)}{k!}=(-1)^k \frac{(-n)\cdots (-n+k-1)}{k!}=(-1)^k {-n+k-1\choose k}.
\end{align}
The definition also applies when $n\leq -1$ and $k=0$, hence 
\begin{align}
    {n\choose 0}=(-1)^0 {-n-1 \choose 0}=1, \text{ if } n\leq -1.
\end{align}
With these notations, it is easy to check that for arbitrary $n\in \Z$ (including $n\leq -1$),
\begin{align}
    T_{d,\omega}^n(x_1,...,x_d)=(x_1+{n\choose 1}\omega, x_2+{n\choose 1}x_1+{n\choose 2}\omega,...,x_d+\sum_{k=1}^{d-1} {n\choose k} x_{d-k}+{n\choose d}\omega)
\end{align}
Recall that $\mathcal{P}_d(x_1,...,x_d)=x_d$ is the projection onto the $d$-th component. Clearly,

\begin{align}    \mathcal{P}_d(T_{d,\omega}^n(x_1,...,x_d))=x_d+\sum_{k=1}^{d-1} {n\choose k} x_{d-k}+{n\choose d}\omega.
\end{align}

The following is the well-known Pascal's triangle formula for $n\geq 0$ and $k\geq 1$.
\begin{align}\label{eq:pascal}
{n\choose k}+{n\choose k-1}={n+1\choose k}.
\end{align}
This formula also holds for $n\leq -1$, and is probably well-known too. We verify below for completeness. 

For $n=-1$ and $k\geq 1$, by \eqref{eq:comb_neg}
\begin{align}
    {-1\choose k}+{-1\choose k-1}=(-1)^k {k\choose k}+(-1)^{k-1}{k-1\choose k-1}=0={0\choose k},
\end{align}

For $n\leq -2$ and $k\geq 1$, by \eqref{eq:comb_neg} and \eqref{eq:pascal} for values $-n+k-1, -n+k-2\geq 0$,
\begin{align}
    {n\choose k}+{n\choose k-1}
    =&(-1)^k {-n+k-1\choose k}+(-1)^{k-1}{-n+k-2\choose k-1}\\
    =&(-1)^k {-n+k-2\choose k}={n+1\choose k}.
\end{align}
This completes the verification.

\section{Spectrum for the electric quantum walk model}\label{sec:spec1}
Proof of Theorem \ref{thm:spec1}. Clearly, for any $\theta\in \T$,
\begin{align}\label{eq:Wo+t}
W_{\omega,\theta,\eta,a,b}=e^{2\pi i\theta} W_{\omega, 0,\eta,a,b},
\end{align}
and hence 
\begin{align}\label{eq:sigW1}
    \sigma(W_{\omega,\theta,\eta,a,b})=e^{2\pi i\theta}\cdot \sigma(W_{\omega,0,\eta,a,b})
\end{align}
Due to minimality of the irrational rotation on $\T$ by $\omega$, we have
\begin{align}\label{eq:sigW2}
    \sigma(W_{\omega,\theta,\eta,a,b})=\sigma(W_{\omega,0,\eta,a,b}).
\end{align}
The claimed result follows by combining \eqref{eq:sigW1} with \eqref{eq:sigW2}.
\qed

\section{Iterated skew-shift on $\T^d$, $d\geq 2$}\label{sec:spec2}
The goal of this section is to prove Theorem \ref{thm:spec2}.
Recall the
\begin{align*}
\mathcal{E}^d_{\omega,a,b}(\vec{x})=\left(\begin{matrix}
&\cdots\\
&be^{-2\pi i \psi_j}  &\underline{0}  &0  &a & & \\
&a^* &0 &0 &-b^*e^{2\pi i \psi_j} & & \\
& & &be^{-2\pi i \psi_{j+1}} &0 &0 &a \\
& & &a^* &0 &0 &-b^*e^{2\pi i \psi_{j+1}} \\
& & & & & &\cdots
\end{matrix}\right),
\end{align*}
where $\underline{0}$ is located at the $(2j,2j)$ position of the matrix, with $\psi_j=\mathcal{P}_d(T_{d,\omega}^{j}(\vec{x}))$.

We will prove $\sigma(\mathcal{E}^d_{\omega,a,b}(\vec{x}))=\partial \D$ for any complex $a,b\in \C$ with $|a|^2+|b|^2=1$ (hence covering more general $a,b$'s than the extended CMV setting), and any vector $\vec{x}\in \T^d$, provided that $\omega$ is irrational.

We first show $\mathcal{E}^d_{\omega,a,b}(\vec{x})$ can be conjugated to a matrix $W_{\beta,a,b}$, of a similar form as \eqref{def:qw}, with certain choices of $\beta=(\beta_j)_{j\in \Z}$:
\begin{align}\label{def:W_beta}
W_{\beta,a,b}=
\left(\begin{matrix}
&\cdots\\
&b e^{2\pi i \beta_{j-1}}  &\underline{0}  &0  &a e^{2\pi i \beta_j} & & \\
&a^* e^{2\pi i \beta_{j-1}}  &0 &0 &-b^* e^{2\pi i \beta_j} & & \\
& & &b e^{2\pi i \beta_j} &0 &0 &a e^{2\pi i \beta_{j+1}} &\\
& & &a^* e^{2\pi i \beta_j} &0 &0 &-b^* e^{2\pi i \beta_{j+1}} &\\
& & & & & &\cdots
\end{matrix}\right),
\end{align}
\begin{lemma}
    There exists matrix $D=\mathrm{diag}(e^{2\pi i\lambda_j})$ such that 
    \begin{align}\label{eq:E_equiv_W_d}
    D\mathcal{E}^d_{\omega,a,b}(\vec{x})D^{-1}=W_{\beta(\vec{x}_-,\omega),a,b},
    \end{align}
    where
    \begin{align}\label{def:beta_j}
        \beta_j(\vec{x}_-,\omega):=-\frac{1}{2}\left(x_{d-1}+\sum_{k=2}^{d-1}{j\choose k-1}x_{d-k}+{j\choose d-1}\omega\right),
    \end{align}
    in which $\vec{x}_-:=(x_1,...,x_{d-1})$.
    In particular \eqref{eq:E_equiv_W_d} implies
    \begin{align}\label{eq:sigma_E=W}
        \sigma(\mathcal{E}^d_{\omega,a,b}(\vec{x}))=\sigma(W_{\beta(\vec{x}_-,\omega),a,b}).
    \end{align}
\end{lemma}
\begin{proof}
   With $D=\mathrm{diag}(e^{2\pi i\lambda_j})$, we have
    \begin{align}
        (D\mathcal{E}^d_{\omega,a,b}(\vec{x})D^{-1})_{2j,2j-1}=&be^{-2\pi i\psi_j}\cdot e^{2\pi i\lambda_{2j}}\cdot e^{-2\pi i\lambda_{2j-1}}=be^{2\pi i(-\psi_j-\lambda_{2j-1}+\lambda_{2j})},\\
        (D\mathcal{E}^d_{\omega,a,b}(\vec{x})D^{-1})_{2j+1,2j-1}=&a^*e^{2\pi i\lambda_{2j+1}}\cdot e^{-2\pi i\lambda_{2j-1}}=a^*e^{2\pi i(\lambda_{2j+1}-\lambda_{2j-1})}\\
        (D\mathcal{E}^d_{\omega,a,b}(\vec{x})D^{-1})_{2j,2j+2}=&ae^{2\pi i\lambda_{2j}}\cdot e^{-2\pi i\lambda_{2j+2}}=ae^{2\pi i(\lambda_{2j}-\lambda_{2j+2})}\\
        (D\mathcal{E}^d_{\omega,a,b}(\vec{x})D^{-1})_{2j+1,2j+2}=&-b^*e^{2\pi i\psi_j}\cdot e^{2\pi i\lambda_{2j+1}}\cdot e^{-2\pi i\lambda_{2j+2}}=-b^*e^{2
        \pi i(\psi_j+\lambda_{2j+1}-\lambda_{2j+2})}.
    \end{align}
    Taking \begin{align}
        \lambda_{2j}=\psi_j/2, \text{ and } \lambda_{2j+1}=-\psi_j/2,
    \end{align} 
    we have
    \begin{align}\label{eq:psi-psi}
        \begin{cases}
            -\psi_j-\lambda_{2j-1}+\lambda_{2j}=-\frac{1}{2}(\psi_j-\psi_{j-1})\\
            \lambda_{2j+1}-\lambda_{2j-1}=-\frac{1}{2}(\psi_j-\psi_{j-1})\\
            \lambda_{2j}-\lambda_{2j+2}=-\frac{1}{2}(\psi_{j+1}-\psi_j)\\
            \psi_j+\lambda_{2j+1}-\lambda_{2j+2}=-\frac{1}{2}(\psi_{j+1}-\psi_j).
        \end{cases}
    \end{align}
    
    Direct computations using the Pascal's triangle (\ref{eq:pascal}) yields:
    \begin{align}\label{def:beta}
        \psi_j-\psi_{j-1}
        =&\sum_{k=1}^{d-1}\left({j\choose k}-{j-1\choose k}\right)x_{d-k}+\left({j\choose d}-{j-1\choose d}\right)\omega \notag\\
        =&\sum_{k=1}^{d-1}{j-1\choose k-1}x_{d-k}+{j-1\choose d-1}\omega \notag\\
        =&x_{d-1}+\sum_{k=2}^{d-1}{j-1\choose k-1}x_{d-k}+{j-1\choose d-1}\omega=-2\beta_{j-1}(\vec{x}_-,\omega),
    \end{align}
    in which we used ${m\choose 0}\equiv 1$ for any $m\in \Z$.
    Combining the above with \eqref{eq:psi-psi}, we have
    \begin{align}
        D\mathcal{E}^d_{\omega,a,b}(\vec{x})D^{-1}=W_{\beta(\vec{x}_-,\omega),a,b}.
    \end{align}
    This is as claimed.
\end{proof}
Now we are in place to prove Theorem \ref{thm:spec2}.
Note that for any $\tau\in \R$, by the minimality of the skew-shift $T_{d,\omega}$, the spectrum set is independent of $\vec{x}\in \T^d$, in particular,
\begin{align}\label{eq:sig_E=sig_E_tau}
    \sigma(\mathcal{E}^d_{\omega,a,b}(\vec{x}))=\sigma(\mathcal{E}^d_{\omega,a,b}(\vec{x}+(0,...,0,\tau,0)),
\end{align}
note $\tau$ is added onto $x_{d-1}$.
It is easy to check by \eqref{def:beta_j}, for any $j\in \Z$, denoting $\vec{x}_-+(0,...,0,\tau)=:\vec{x}_{-}^{(\tau)}$, that
\begin{align}\label{eq:beta_j_shift_tau}
    \beta_j(\vec{x}_{-}^{(\tau)}),\omega)=\beta_j(\vec{x}_-,\omega)-\frac{\tau}{2}.
\end{align}
Hence,  by \eqref{def:W_beta} and \eqref{eq:beta_j_shift_tau},
\begin{align}
    W_{\beta(\vec{x}_{-}^{(\tau)},\omega),a,b}=e^{-\pi i\tau} W_{\beta(\vec{x}_-,\omega),a,b},
\end{align}
which implies
\begin{align}
    \sigma(W_{\beta(\vec{x}_-^{(\tau)},\omega),a,b})=e^{-\pi i\tau} \cdot \sigma(W_{\beta(\vec{x}_-,\omega),a,b}).
\end{align}
Combining the above with \eqref{eq:sigma_E=W}, we have for any $\vec{x}\in \T^d$ that 
\begin{align}
    \sigma(\mathcal{E}^d_{\omega,a,b}(\vec{x}+(0,...,0,\tau,0))=e^{-\pi i\tau} \cdot \sigma(\mathcal{E}^d_{\omega,a,b}(\vec{x})).
\end{align}
This further combined with \eqref{eq:sig_E=sig_E_tau} implies
\begin{align}\label{eq:sig_E_invariant}
    \sigma(\mathcal{E}^d_{\omega,a,b}(\vec{x}))=e^{-\pi i\tau} \cdot \sigma(\mathcal{E}^d_{\omega,a,b}(\vec{x})), \text{ for any } \tau \in \R.
\end{align}
Since $|a|^2+|b|^2=1$, $\mathcal{E}_{\omega,a,b}(\vec{x})$ is unitary, which implies $\sigma(\mathcal{E}_{\omega,a,b}(\vec{x}))\subset \partial \D$. Combining this with \eqref{eq:sig_E_invariant} yields
\begin{align}
    \sigma(\mathcal{E}^d_{\omega,a,b}(\vec{x}))=\partial \D.
\end{align}
\qed

\bibliographystyle{amsplain}

\end{document}